\documentclass[a4paper]{article}
\usepackage[T1]{fontenc}

\usepackage{amsmath,amssymb,amsthm}
\usepackage{microtype}
\usepackage{xcolor}
\usepackage{subcaption}
\usepackage{longtable,booktabs,array}
\usepackage{boxedminipage}
\usepackage{framed}
\usepackage{xifthen}
\usepackage{tabularx}
\usepackage{tikz}
\usetikzlibrary{calc,math,patterns,shapes,backgrounds}
\tikzstyle{path} = [color=black,opacity=.30,line cap=round, line join=round, line width=10pt]
\usepackage[hidelinks]{hyperref}
\usepackage{cleveref}

\newtheorem{lemma}{Lemma}
\newtheorem{theorem}{Theorem}

\newlength{\RoundedBoxWidth}
\newsavebox{\GrayRoundedBox}
\newenvironment{GrayBox}[1]{\setlength{\RoundedBoxWidth}{.93\textwidth}
    \def\boxheading{#1}
    \begin{lrbox}{\GrayRoundedBox}
       \begin{minipage}{\RoundedBoxWidth}}{   \end{minipage}
    \end{lrbox}
    \begin{center}
    \begin{tikzpicture}\node(Text)[draw=black!20,fill=white,rounded corners,inner sep=2ex,text width=\RoundedBoxWidth]{\usebox{\GrayRoundedBox}};
        \coordinate(x) at (current bounding box.north west);
        \node [draw=white,rectangle,inner sep=3pt,anchor=north west,fill=white]
        at ($(x)+(6pt,.75em)$) {\boxheading};
    \end{tikzpicture}
    \end{center}}

\newenvironment{defproblemx}[2][]{\noindent\ignorespaces \FrameSep=6pt\parindent=0pt\vspace*{-1.5em}
                \ifthenelse{\isempty{#1}}{\begin{GrayBox}{\textsc{#2}}}{\begin{GrayBox}{\textsc{#2}  parameterized by~{#1}}}
                \begin{tabular*}{\textwidth}{@{\hspace{.1em}} >{\itshape} p{1.8cm} p{0.8\textwidth} @{}}}{
                \end{tabular*}\end{GrayBox}\ignorespacesafterend
            }

\newcommand{\defproblema}[3]{\begin{defproblemx}{#1}
    Input:  & #2 \\
    Task: & #3
  \end{defproblemx}
}

\newtheorem{redrule}{Reduction rule}

\newcommand{\NP}{\textsf{\textup{NP}}}
\newcommand{\FPT}{\textsf{\textup{FPT}}}
\newcommand{\APX}{\textsf{\textup{APX}}}
\DeclareMathOperator{\poly}{poly}

\newcommand{\ver}[2]{\ensuremath{#1_{\vert #2}}}

\newcommand{\BEVS}{\textsc{Bicluster Editing with Vertex Splitting}}
\newcommand{\BEOVS}{\textsc{Bicluster Editing with One-Sided Vertex Splitting}}

\usepackage{etoolbox}

\title{Overlapping Biclustering}

\usepackage{authblk}
\author[1,2]{Matthias Bentert}
\author[1]{Pål Grønås Drange}
\author[1]{Erlend Haugen}
\affil[1]{University of Bergen, Norway}
\affil[2]{Technische Universität Berlin, Germany}

\date{}

\begin{document}

\maketitle

\begin{abstract}
    We study the problem of transforming bipartite graphs into bicluster graphs. Abu-Khzam, Isenmann, and Merchad~[IWOCA '25] introduced two variants of this problem.
In both problems, the goal is to transform a bipartite graph into a bicluster graph with at most~$k$ operations, where the allowed operations are inserting an edge, deleting an edge, and splitting a vertex.
    Splitting a vertex~$v$ refers to replacing~$v$ by two new vertices  whose combined neighborhood equals the neighborhood of~$v$.
    The latter models overlapping clusters, that is, vertices belonging to multiple clusters, and is motivated by several real-world applications.
    The versions differ in that one variant allows splitting any vertex, while the second variant only allows vertex splits on one side of the bipartition.

    Regarding computational complexity, they showed APX-hardness for both variants and a polynomial kernel (with~$O(k^5)$ vertices) for the one-sided variant.
    They asked as open problems whether the polynomial kernel can be improved and whether it can also be extended for the other variant.
    We answer both questions in the affirmative and give kernels with~$O(k^2)$ vertices for both variants.
    We also show that both problems can be solved in~$O(k^{11k} + n + m)$ time, where $n$ and $m$ denote the number of vertices and edges in the input graph, respectively.
\end{abstract}

\section{Introduction}

Clustering is a fundamental task in data analysis, aiming to group similar
objects based on a given similarity measure. In \emph{correlation clustering},
the input is a set of objects with pairwise similarity or dissimilarity
information, and the goal is to partition the objects into clusters of mutually
similar entities~\cite{bansal_correlation_2004}. This approach is widely used
in machine learning, document classification, gene expression analysis, image
segmentation, and social network analysis.
Clustering is often modeled as the graph problem \textsc{Cluster Editing}~\cite{crespelle_survey_2023,komusiewicz_cluster_2012}.
Here, objects are vertices and an edge between two vertices indicates similarity.
The task is then to add or delete the fewest number of edges to obtain a \emph{cluster graph}---a disjoint union of cliques.

However, many datasets do not naturally express similarity in terms of pairwise
relationships. Instead, they are \emph{dichotomous}, representing binary
associations between two distinct types of objects. These datasets are common
in a variety of applications: politicians and the bills they support,
movie enjoyers and the films they watch, or students and the courses they
take~\cite{jiang_bifold_2017}. Such data can be modeled as bipartite graphs,
where one part corresponds to decision-makers and the other to their
choices. Similar binary structures arise in gene-disease
associations~\cite{bauer-mehren_gene-disease_2011} and chemical
reactions~\cite{craciun_multiple_2006}.

To cluster such dichotomous data, cliques are no longer the appropriate
target. Instead, the goal becomes transforming the input into a \emph{bicluster
  graph}---a disjoint union of bicliques (complete bipartite graphs). This
leads to the problem of \textsc{Bicluster Editing}, which asks for the minimum
number of edge modifications needed to achieve such a
structure. \textsc{Bicluster Editing} has numerous applications, including gene
expression analysis, community detection, and in recommendation systems~\cite{barber2007modularitycommunity,cheng1999biclusteringexpression,MO04,pontes2015biclusteringexpression,sun_efficient_2014,tanay2005biclusteringalgorithms}.

In certain cases, some objects naturally belong to multiple clusters.  In
document classification, where we assign keywords to documents, a single
document often fits several keywords, and the resulting groups do not form
disjoint clusters.  In community detection within social network analysis, we
typically assume each person belongs to only one community---an overly strong
assumption.  Similarly, in sentiment analysis, where we classify the emotional
tone of a message, a sentence may naturally belong to more than one category.
For example ``\textit{I love this example, but I hate how long it took me to
  come up with it}.''

In the bipartite case, similar challenges arise.  Consider movie recommender systems.
On the one side, we have the users and on the other side, we have the movies.
Edges indicate that a particular user likes a particular movie.
One can imagine a group of users all liking the same set of action movies---forming a biclique---and another group doing the same for romantic films.
But full separation is unlikely since some users will enjoy both action movies and romances.
Naturally, such users should belong to both clusters, that is, clusters should allow for some amount of \emph{overlap}~\cite{du_overlapping_2008}.
Such overlapping (bi)clusters are modeled in the context of graph problems by \emph{vertex splitting}: a vertex~$v$ with neighborhood~$N(v)$ is split into (replaced by) two vertices~$v_1$ and~$v_2$ such
that~$N(v_1) \cup N(v_2) = N(v)$, where each new vertex represents the involvement of~$v$ in one (bi)cluster (or is split further).
An example of such a vertex split is given in \Cref{fig:splitting-example}.
Allowing this operation in addition to the insertion and removal of edges leads to the following problem, first studied by Abu-Khzam et al.~\cite{abu-khzam2025biclusterediting}.
\begin{figure}[t]
\centering
    \begin{tikzpicture}
        \node[draw, circle] (A1) at (0, 1) {$a_1$};
        \node[draw, circle] (A2) at (0, 2) {$a_2$};
        \node[draw, circle] (A3) at (0, 3) {$a_3$};
        \node[draw, circle] (A4) at (0, 4) {$a_4$};
        \node[draw, circle] (A5) at (0, 5) {$a_5$};
        \node[draw, circle] (B1) at (3, 1) {$b_1$};
        \node[draw, circle] (B2) at (3, 2) {$b_2$};
        \node[draw, circle] (B4) at (3, 4) {$b_4$};
        \node[draw, circle] (B5) at (3, 5) {$b_5$};

\draw[gray] (A1) -- (B1);
        \draw[gray] (A1) -- (B2);
        \draw[gray] (A2) -- (B1);
        \draw[gray] (A2) -- (B2);
        \draw[gray] (A3) -- (B1);
        \draw[gray] (A3) -- (B2);
        \draw[gray] (A4) -- (B4);
        \draw[gray] (A5) -- (B5);
        \draw[gray] (A3) -- (B4);
        \draw[gray] (A3) -- (B5);
        \draw[gray] (A4) -- (B5);
        \draw[gray] (A5) -- (B4);
    \end{tikzpicture}
    \hspace*{2cm}
\begin{tikzpicture}
        \node[draw, circle] (A1) at (0, 1) {$a_1$};
        \node[draw, circle] (A2) at (0, 2) {$a_2$};
        \node[draw, circle] (A3) at (0, 3) {\scriptsize$a_3^1$};
        \node[draw, circle] (A3p) at (1.4, 3) {\scriptsize$a_3^2$};
        \node[draw, circle] (A4) at (0, 4) {$a_4$};
        \node[draw, circle] (A5) at (0, 5) {$a_5$};
        \node[draw, circle] (B1) at (3, 1) {$b_1$};
        \node[draw, circle] (B2) at (3, 2) {$b_2$};
        \node[draw, circle] (B4) at (3, 4) {$b_4$};
        \node[draw, circle] (B5) at (3, 5) {$b_5$};

\draw[gray] (A1) -- (B1);
        \draw[gray] (A1) -- (B2);
        \draw[gray] (A2) -- (B1);
        \draw[gray] (A2) -- (B2);
        \draw[gray] (A3) -- (B1);
        \draw[gray] (A3) -- (B2);
        \draw[gray] (A4) -- (B4);
        \draw[gray] (A5) -- (B5);
        \draw[gray] (A3p) -- (B4);
        \draw[gray] (A3p) -- (B5);
        \draw[gray] (A4) -- (B5);
        \draw[gray] (A5) -- (B4);
    \end{tikzpicture}
\caption{An example instance on the left where two operations are needed when only edge insertions and edge removals are allowed. On the right side a solution resulting from a single vertex split is shown.}
  \label{fig:splitting-example}
\end{figure}

\defproblema{Bicluster Editing with Vertex Splitting}{ A bipartite graph $G = (V_1 \uplus V_2, E)$ and an integer $k\geq 0$.}{ Decide whether there exists a bicluster graph which is the result of at most~$k$~edit operations on $G$, where each operation adds an edge, removes an edge, or splits a vertex. }

We mention that there are two natural variants of the \textsc{Bicluster Editing} problem and it makes sense to study both with the additional vertex splitting operation.
In one variant, the input is an arbitrary graph, and the goal is to transform it into a bicluster graph without constraints on which vertices can belong to either side of the resulting bipartite graph.
This model fits applications such as ownership or trading networks, where entities may
simultaneously act on both sides (e.g., as buyers and sellers).
In the second variant, the input graph is bipartite with a fixed bipartition.
This assumption naturally occurs in the aforementioned application in recommender systems, but also in the a context of phylogenetics and other biological datasets.
We focus on the second variant in this paper as we believe that it has the potential to improve recommender systems, especially in terms of explainability.
While many current systems rely
on statistical approaches or machine learning methods and operate as black boxes,
clustering-based approaches can reveal groupings transparently, providing both users and providers with insight into why certain recommendations are made.

We also briefly mention the following problem, introduced by Abu-Khzam, Isenmann, and Merchad~\cite{abu-khzam2025biclusterediting}, as both our kernel and \FPT{} algorithms work for this:

\defproblema{Bicluster Editing with One-Sided Vertex Splitting}{ A bipartite graph $G = (V_1 \uplus V_2, E)$ and an integer $k\geq 0$.}{ Decide whether there exists a bicluster graph which is the result of at most~$k$~edit operations on $G$, where each operation adds an edge, removes an edge, or splits a vertex in $V_1$.}

\paragraph{Related Work.}
We start with related work on vertex splitting.
One of the first uses of vertex splitting as a graph modification operation appeared in the context of graph drawing as a way to make graphs planar.
Eades and de Mendon{\c{c}}a Neto~\cite{EM95} introduced this problem under the name  \textsc{Planar Vertex Splitting}.
Faria et al.~\cite{FFN01} showed that it is \NP-complete, even on cubic graphs.
Ahmed et al.~\cite{AKK22} showed that vertex splitting to bipartite graphs is
\FPT{} and used this result to develop graph drawing algorithms~\cite{AAB+23}.
Vertex splitting was also studied in the context of graph classes of bounded treewidth: Baumann et al.~\cite{BPR24} showed that for MSO$_2$ testable graph classes $\Pi$ with bounded treewidth, vertex splitting to~$\Pi$ is \FPT{}.

Abu-Khzam et al.~\cite{AABD+23} used vertex splitting as an added operation in \textsc{Cluster Editing} to model overlapping clustering. Here, the idea
is that a vertex that belongs to two different clusters can be split into two,
one copy for each of the two clusters it belongs to.
The resulting problem as well as different variants of \textsc{Cluster Editing} with this additional vertex splitting operation have since been studied~\cite{abu-khzam2026complexity2club,BCDRS24,firbas_complexity_2025}.

We continue with related work on \textsc{Bicluster Editing}.
The problem is \NP-complete~\cite{amit2004biclustergraph}, even
on subcubic graphs~\cite{drange2015fastbiclustering}, and has been extensively
studied due to its applications in computational biology~\cite{MO04}.  The
problem has seen steady progress in both kernelization and fixed-parameter
tractability: Protti et al.~\cite{protti_applying_2009} gave an $O(k^2)$ vertex
kernel and a branching algorithm running in $O(4^k + n + m)$ time.  This was
improved by Guo et al.~\cite{GHKZ08} to $O(3.24^k + n + m)$, along with a
smaller kernel, which was later found to contain a minor
flaw~\cite{lafond_even_2020}.  Subsequent works further advanced the
understanding and efficiency of the problem~\mbox{\cite{drange2015fastbiclustering,sun_efficient_2014,XK22}}.
These efforts culminated in the current state-of-the-art results by
Tsur~\cite{Tsur23}, who presented an \FPT{} algorithm running in
$O^*(2.22^k)$ time, and Lafond~\cite{Lafond24} who gave a refined kernel containing at most $4.5k$
vertices.  For a comprehensive overview of early developments and biological
motivations, we refer to the survey by Madeira and Oliveira~\cite{MO04}.

Abu-Khzam, Isenmann, and Merchad~\cite{abu-khzam2025biclusterediting} recently introduced the problems \textsc{Bicluster Editing with Vertex Splitting} and its one-sided variant \BEOVS.
They proved both problems to be \NP-complete, \APX-hard, and intractable under ETH even on bipartite planar graphs of maximum degree three.
On the positive side, they showed that \BEOVS{} admits a kernel with~$O(k^5)$ vertices that can be computed in~$O(n^2)$ time and developed an algorithm running in~$O(5^{3k}k^{2k}n^2)$ for it.

\paragraph{Our Contribution.}
We show that both problems admit kernels with~$O(k^2)$ vertices that can be computed in~$O(n + m)$~time.
We also show that the algorithm by Abu-Khzam, Isenmann, and Merchad for \BEOVS{} contains a flaw and design an algorithm running in~${O(k^{11k} + n + m)}$~time for both \BEVS{} and \BEOVS. These two results solve two open problems raised by
Abu-Khzam, Isenmann, and Merchad~\cite{abu-khzam2025biclusterediting}.

\section{Preliminaries}
\label{sec:prelim}

For a positive integer~$k$, we use~$[k]$ to denote the set~$\{1,2,\ldots,k\}$ of all positive integers up to~$k$.
All logarithms in this paper use 2 as their base.
We use standard graph-theoretic notation and refer the reader to the textbook by Diestel~\cite{diestel2005graphtheory} for standard definitions.  For $G = (V, E)$, we will denote by $n_G$ and $m_G$ the number of vertices and edges, respectively.
All graphs in this work are simple, unweighted, and undirected.
We denote the open neighborhood of a vertex~$v$ by~$N_G(v)$.  When~$G$ is clear from context, we will drop the subscripts.
For an introduction to parameterized complexity, fixed-parameter tractability, and kernelization, we refer the reader to the textbooks by Flum and Grohe~\cite{flum2006parameterizedcomplexity}, Niedermeier~\cite{Niedermeier06}, and Cygan et al.~\cite{CFKLMPPS15}.

A graph~$G=(V_1 \uplus V_2,E)$ is bipartite if each edge in~$E$ has one endpoint in~$V_1$ and one endpoint in~$V_2$---note that any one of those can be empty.
A biclique is a complete bipartite graph, that is, each vertex in~$V_1$ is adjacent to each vertex in~$V_2$.
If~$G=(V_1 \uplus V_2,E)$ is a biclique, then we will also refer to~$G=(V_1,V_2)$ as a biclique for notational convenience.
A \emph{bicluster graph} is a graph in which each connected component is a biclique.
The class of bicluster graphs is exactly the graphs that have neither
$C_3$ nor $P_4$ as an induced subgraph. An (inclusive) vertex split replaces a vertex~$v$ by two vertices~$v_1,v_2$ such that~$N(v_1) \cup N(v_2) = N(v)$.
An example of the splitting operation is given in \Cref{fig:incsplit} and we call the two new vertices~$v_1$ and~$v_2$ \emph{copies} of the original vertex~$v$ (and if~$v_1$ or~$v_2$ are further split in the future, then the resulting vertices are also called copies of~$v$).
For the sake of notational convenience, we also say that~$v$ is a copy of itself.
\begin{figure}[t]
  \centering
  \begin{tikzpicture}[scale=0.9]
        \node[circle, draw, label=left:$a_1$] at (-1,1) (a1) {};
        \node[circle, draw, label=left:$a_2$] at (-1,0) (a2) {};
        \node[circle, draw, label=right:$b_1$] at (1,1) (b1) {} edge(a1) edge(a2);
        \node[circle, draw, label=right:$b_2$] at (1,0) (b2) {} edge(a1) edge(a2);
        \node[circle, draw, label=$c$] at (0,2) (c) {} edge(a1) edge(a2) edge(b1) edge(b2);
        \node[circle, draw, label=below:$v$] at (0,-1) (v) {} edge(a1) edge(a2) edge(b1) edge(b2) edge(c);

        \node at(2.5,.5) {\LARGE $\Rightarrow$};

        \node[circle, draw, label=left:$a_1$] at (4,1) (a1) {};
        \node[circle, draw, label=left:$a_2$] at (4,0) (a2) {};
        \node[circle, draw, label=right:$b_1$] at (6,1) (b1) {} edge(a1) edge(a2);
        \node[circle, draw, label=right:$b_2$] at (6,0) (b2) {} edge(a1) edge(a2);
        \node[circle, draw, label=$c$] at (5,2) (c) {} edge(a1) edge(a2) edge(b1) edge(b2);
        \node[circle, draw, label=left:$v_1$] at (4,-1) (v1) {} edge(b1) edge(b2) edge(c);
        \node[circle, draw, label=right:$v_2$] at (6,-1) (v2) {} edge(a1) edge(a2) edge(c);

        \node at(7.5,.5) {\LARGE $\Rightarrow$};

        \node[circle, draw, label=left:$a_1$] at (9,1) (a1) {};
        \node[circle, draw, label=left:$a_2$] at (9,0) (a2) {};
        \node[circle, draw, label=right:$b_1$] at (11,1) (b1) {} edge(a1) edge(a2);
        \node[circle, draw, label=right:$b_2$] at (11,0) (b2) {} edge(a1) edge(a2);
        \node[circle, draw, label=left:$c_1$] at (9,2) (c1) {} edge(b1) edge(b2);
        \node[circle, draw, label=right:$c_2$] at (11,2) (c2) {} edge(a1) edge(a2);
        \node[circle, draw, label=left:$v_1$] at (9,-1) (v1) {} edge(b1) edge(b2) edge(c2);
        \node[circle, draw, label=right:$v_2$] at (11,-1) (v2) {} edge(a1) edge(a2) edge(c1);
    \end{tikzpicture}
    \caption{An illustration of two (inclusive) vertex splits (in a
      non-bipartite graph).  In the first split, the vertex~$v$ is replaced
      by~$v_1$ and~$v_2$, with the vertices~$a_1,a_2,b_1$ and~$b_2$ being
      adjacent to exactly one of the two vertices and~$c$ being adjacent to
      both.  In the second, $c$ is split into~$c_1$ and~$c_2$.}
    \label{fig:incsplit}
\end{figure}
An edit sequence is a sequence of operations, where each operation is (i) the addition of an edge, (ii) the removal of an edge, or (iii) the inclusive split of a vertex.
We denote the graph resulting from applying an edit sequence~$\sigma$ to a graph~$G$ by~$\ver{G}{\sigma}$.

A \emph{critical independent set} is a maximal set~$I$ of vertices such that~$N(u)=N(v)$ for all~$u,v \in I$.
Equivalently, two vertices~$u$ and~$v$ belong to the same critical independent set if and only if they are false twins.
It is known that critical independent sets form a partition of the vertex set of a graph and that all critical independent sets can be computed in linear time \cite{HM91}.
The \emph{critical independent set quotient graph}~$\mathcal I$ of~$G$ contains a node for each critical independent set in~$G$ and two nodes are adjacent if and only if the two respective critical independent sets $I_1$ and $I_2$ are adjacent, that is, there is an edge between each vertex in~$I_1$ and each vertex in~$I_2$.
Note that by the definition of critical independent sets, this is equivalent to the condition that at least one edge~$\{u,v\}$ with~$u \in I_1$ and~$v \in I_2$ exists.
To avoid confusion, we will call the vertices in~$\mathcal{I}$ \emph{nodes} and in~$G$ \emph{vertices}.

\section{Structural Results}

In this section, we prove two structural results about optimal solutions to
\BEVS{} and \BEOVS{} that are used in later proofs.
The first state that the bipartitions of the final bicluster graph follow the bipartition of the input graph, that is, no edges between copies of vertices in~$V_1$ and~$V_2$ are added.

\begin{lemma}
    \label{lem:noconversion}
    Let $(G=(V_1 \uplus V_2,E), k)$ be a yes-instance of \BEVS{} or \BEOVS.
    Then, there exists a solution~$\sigma$ of length at most~$k$ such that for each biclique~$H_j=(A_j,B_j)$ in $\ver{G}{\sigma}$ it holds that~$A'_j \subseteq V_1$ and~$B'_j \subseteq V_2$, where~$A'_j$ and~$B'_j$ consist of all vertices~$v \in V_1 \uplus V_2$ such that~$A_j$ and $B_j$ contain a single copy of~$v$, respectively.
\end{lemma}
\begin{proof}
    Let~$\sigma$ be a solution of size~$k' \leq k$.
    We will construct a solution~$\sigma'$ of size~$k'' \leq k'$ that satisfies the listed requirements.
    If~$\sigma$ already satisfies the requirements, then we are done.
    So assume that it does not and let~$H_j=(A_j,B_j)$ be a biclique in~$\ver{G}{\sigma}$ that does not satisfy the requirement.
    Then, we assume without loss of generality that~$A'_j \cap V_1 \neq \emptyset$, $B'_j \cap V_2 \neq \emptyset$, and~$B'_j \cap V_1 \neq \emptyset$ (the case where~$A'_j \cap V_2 \neq \emptyset$ is symmetric even for \BEOVS).
    We partition~$B_j$ into two sets~$B_1$ and~$B_2$ such that~$B_1$ contains all copies of vertices in~$V_1$ and~$B_2$ contains all copies of vertices in~$V_2$.
    We also partition~$A_j$ into~$A_1$ and~$A_2$, where~$A_2$ is potentially empty.
    Note that all edges between vertices in~$A_1$ and~$B_1$ and all edges between vertices in~$A_2$ and~$B_2$ are added in~$\sigma$ since~$V_1$ and~$V_2$ induce independent sets in~$G$.
    Consider the sequence~$\sigma'$ that is equal to~$\sigma$ except for the fact that it does not create these edges.
    Note that~$\sigma'$ is strictly shorter than~$\sigma$ so it only remains to show that~$\sigma'$ is a valid solution.

    To this end, consider the effect on~$H_j = (A_j,B_j)$.
    All edges between vertices in~$A_1$ and in~$B_2$ are still there.
    So are all edges between vertices in~$A_2$ and in~$B_1$.
    Since~$H_j$ is a biclique, there are no edges between vertices in~$A_1$ and in~$A_2$ or between vertices in~$B_1$ and in~$B_2$.
    Moreover, by construction, there are no edges between vertices in~$A_1$ and in~$B_1$ or between vertices in~$A_2$ and in~$B_2$ in~$\ver{G}{\sigma'}$.
    Thus, the vertices in~$H_j$ now form the two bicliques~$H_j^1=(A_1,B_2)$ and~$H_j^2 = (B_1,A_2)$.
    Since~$H_j$ was chosen arbitrarily, the same holds for all bicliques that did not initially satisfy the requirement.
    Moreover, since~$A'_1,B'_1 \subseteq V_1$ and~$A'_2,B'_2 \subseteq V_2$ by construction, all new bicliques also satisfy the requirement of the lemma.
    This concludes the proof.
\end{proof}

\bigskip

The following lemma is a generalization of a lemma due to Guo et al.~\cite{GHKZ08} in the context of vertex splitting.
It states that two copies of the same vertex are always contained in two different bicliques in (the graph obtained by) any optimal solution.

\begin{lemma}
\label{lem:cis}
  Let $(G=(V_1 \uplus V_2,E), k)$ be a yes-instance of \BEVS{} or \BEOVS.
  Then, there exists a solution~$\sigma$ of length at most~$k$ such that for each critical independent set~$I_i \subseteq V_1$ in~$G$ and each biclique~$H_j=(A_j,B_j)$ in~$\ver{G}{\sigma}$ it holds that~$B_j$ does not contain a copy of any vertex in~$I_i$ and~$A_j$ contains exactly one copy of each vertex in~$I_i$ or does not contain any copy of a vertex in~$I_i$.
  The same applies to critical independent sets~$I_p \subseteq V_2$ with the roles of~$A_j$ and~$B_j$ swapped.
\end{lemma}

\begin{proof}
  Let $\hat\sigma$ be a solution for~$(G,k)$ that satisfies \Cref{lem:noconversion}.
  For each critical independent set $I_i$, we select a representative vertex $r_i \in I_i$ by picking any vertex in~$I_i$ with the fewest number of appearances in $\hat\sigma$, that is, a vertex which minimizes the sum of incident added edges, incident removed edges, and vertex splits in~$\hat\sigma$.
  Note that due to the assumption that~$\hat\sigma$ satisfies \Cref{lem:noconversion}, it holds for any critical independent set~$I_i \subseteq V_1$, any critical independent set~$I_{i'}$, and each biclique~$H=(A_j,B_j)$ in~$\ver{G}{\hat\sigma}$ that~$A_j$ does not contain a copy of~$I_{i'}$ and~$B_j$ does not contain a copy of a vertex in~$I_i$.
  If it also holds for each critical independent set~$I_i \subseteq V_1$, each critical independent set~$I_{i'} \subseteq V_2$, and each biclique~$H=(A_j,B_j)$ in~$\ver{G}{\hat\sigma}$ that (i)~$A_j$ does not contain any copy of a vertex in~$I_i$ or~$A_j$ contains exactly one copy of each vertex in~$I_i$ and (ii)~$B_j$ does not contain any copy of a vertex in~$I_{i'}$ or~$B_j$ contains exactly one copy of each vertex in~$I_{i'}$, then~$\hat\sigma$ satisfies all requirements of the lemma statement, and we are done.

  Otherwise, there exists a biclique $H_j=(A_j,B_j)$ in~$\ver{G}{\hat\sigma}$ such that~(i)~$A_j$ contains two copies of some vertex~$v \in I_i$ or (ii) there exists a critical independent set~$I_i$ (or~$I_{i'}$) and two vertices~$u,v$ in it such that~$A_j$ ($B_j$) contains a copy of~$u$ but no copy of~$v$.
  In case (i), we remove all operations from~$\hat\sigma$ involving one of the two copies.
  Note that this results in a solution of strictly shorter length as we do not need to create this additional copy and bicliques are closed under vertex deletion.
  In case (ii), we assume without loss of generality that~$A_j$ contains a copy of~$u$ but no copy of~$v$ as the case where~$B_j$ contains a copy of~$u$ but no copy of~$v$ is analogous.
  Then, there also exists such a pair of vertices where one of~$u$ and~$v$ is~$r_i$: if~$A_j$ contains a copy of~$r_i$, then we can take the pair~$(r_i,v)$ and if~$A_j$ does not contain a copy of~$r_i$, then we can take the pair~$(u,r_i)$.
  Let~$w \in \{u,v\}$ be the other vertex.

  We find a new optimal solution by removing all operations involving~$w$ and copying all operations including $r_i$ and replacing $r_i$ by $w$ in the copy.
  For the sake of notational convenience, we say that if some operation~$e_p$ involves the~$j$\textsuperscript{th} copy of~$r_i$, then the operation~$e_{p+1}$ is a copy of~$e_p$ and it uses the $j$\textsuperscript{th} copy of~$w$.
  Moreover, when splitting another vertex (not $r_i$ or $w$), then we treat each copy of $w$ the same as the corresponding copy of~$r_i$.
  Note that since~$\hat\sigma$ satisfies \Cref{lem:noconversion}, it does not create any edge between a copy of~$w$ and a copy of~$r_i$.
  Our modifications to~$\hat\sigma$ also do not add such edges so the resulting sequence~$\sigma'$ still satisfies \Cref{lem:noconversion}.
  Since $w$ is involved in at least as many operations as~$r_i$ (recall that $r_i$ was picked as having the fewest number of appearances in~$\hat\sigma$), it holds that~$\sigma'$ will be at most as long as~$\hat\sigma$.

  We next show that~$\sigma'$ is also a solution.
  Note that the graph~$\ver{G}{\sigma'}$ is the graph obtained from~$\ver{G}{\hat\sigma}$ by deleting all copies of~$w$, and then adding a false twin for each copy of~$r_i$.
  Since the class of bicluster graphs is closed under vertex deletion and under adding false twins, it follows that~$\ver{G}{\sigma'}$ is also a bicluster graph.
  Moreover, the number of vertices that behave differently than their representative is one smaller in~$\sigma'$ compared to~$\hat\sigma$.
  Hence, repeating the above procedure at most~$n$ times results in an optimal solution~$\sigma$ as stated in the lemma.
\end{proof}

In the following section, we will use \Cref{lem:cis} to develop a polynomial kernel, as well as a~$2^{O(k \log k)}(n+m)$-time algorithm for both problem variants, and importantly, the kernelization algorithms run in linear time in the input.

\section{A Polynomial Kernel}
\label{sec:kernel}

In this section, we prove that \BEVS{} and \BEOVS{} admit kernels with~$O(k^2)$ vertices which can be computed in linear time.
This improves upon the kernel with~$O(k^5)$ vertices computable in~$O(n^2)$ time by Abu-Khzam, Isenmann, and Merchad~\cite{abu-khzam2025biclusterediting}.
We note that Lafond recently proved a $4.5k$-vertex kernel for \textsc{Bicluster Editing}~\cite{Lafond24}.
However, the kernel does not generalize to \BEVS{} or \BEOVS.
The kernel consists of three rules and then an argument that if the resulting graph contains more than~$4.5k$ vertices, then it is a no-instance.
The graph in \Cref{fig:counterexample} does not allow for the application of any of the three rules of Lafond for~$k=2$ but it is a yes-instance since the two vertices~$u,v$ can each be split once to obtain a bicluster graph.

\begin{figure}[t]
  \centering
  \begin{subfigure}[t]{.48\textwidth}
    \centering
    \begin{tikzpicture}
        \node[circle, draw] at (-1,3.5) (a) {};
        \node[circle, draw] at (-1,3) (b) {};
        \node[circle, draw,label=left:$u$] at (-1,2) (u) {};
        \node[circle, draw,label=left:$v$] at (-1,1.5) (v) {};
        \node[circle, draw] at (-1,.5) (c) {};
        \node[circle, draw] at (-1,0) (d) {};

        \node[circle, draw] at (1,2.75) {} edge(a) edge(b) edge(u) edge(v);
        \node[circle, draw] at (1,2.25) {} edge(a) edge(b) edge(u) edge(v);
        \node[circle, draw] at (1,1.25) {} edge(c) edge(d) edge(u) edge(v);
        \node[circle, draw] at (1,.75) {} edge(c) edge(d) edge(u) edge(v);
    \end{tikzpicture}
    \caption{A graph in which none of the reduction rules used by Lafond~\cite{Lafond24} apply which is a yes-instance for~$k=2$ but contains more than~$9$ vertices.}
    \label{fig:counterexample}
\end{subfigure}\hfill \begin{subfigure}[t]{.48\textwidth}
  \centering
    \begin{tikzpicture}
        \node[circle, draw] at (-1,3.5) (a1) {};
        \node[circle, draw] at (-1,3) (a2) {};
        \node[circle, draw] at (-1,2.5) (a3) {};
        \node[circle, draw] at (-1,2) (a4) {};
        \node[circle, draw] at (-1,1.5) (a5) {};
        \node[circle, draw] at (-1,1) (a6) {};
        \node[circle, draw] at (-1,.5) (a7) {};
        \node[circle, draw] at (1,3.5) (b1) {};
        \node[circle, draw] at (1,3) (b2) {};
        \node[circle, draw] at (1,2.5) (b3) {};
        \node[circle, draw] at (1,2) (b4) {};
        \node[circle, draw] at (1,1.5) (b5) {};
        \node[circle, draw] at (1,1) (b6) {};
        \node[circle, draw] at (1,.5) (b7) {};
        \node[circle, draw] at (-1,-.5) (c1) {};
        \node[circle, draw] at (-1,-1) (c2) {};
        \node[circle, draw] at (-1,-1.5) (c3) {};
        \node[circle, draw] at (1,-.5) (d1) {};
        \node[circle, draw] at (1,-1) (d2) {};
        \node[circle, draw] at (1,-1.5) (d3) {};
        \node at(-2,2) {$A$};
        \node at(2,2) {$B$};
        \foreach \i in {1,2,...,7}{
            \foreach \j in {1,2,...,7}{
                \draw (a\i) -- (b\j);
            }
        }
        \foreach \i in {1,2,...,7}{
            \foreach \j in {1,2,3}{
                \draw (a\i) -- (d\j);
                \draw(c\j) -- (b\i);
            }
        }
    \end{tikzpicture}
    \caption{A graph in which removing a vertex from either set~$A$ or~$B$ reduces the size of an optimal solution. An optimal solution splits each vertex in~$A$ once and has therefore size~$7$. Removing one vertex from~$A$ or~$B$ drops the solution size to~$6$ as only the vertices in the smaller set need to be split.}
    \label{fig:counterexample2}
  \end{subfigure}
  \caption{Examples showing that previous kernelization techniques for \textsc{Bicluster Editing} do not extend to our setting.}
\end{figure}

Moreover, a common approach for linear kernels for \textsc{Bicluster Editing} is to argue that if the second neighborhood\footnote{The second neighborhood of $v$, $N^2(v)$ is the set of vertices at distance 2 from $v$.} of a critical independent set~$I_i$ is smaller than~$|I_i|$, then one can safely remove a vertex from~$I_i$.
The graph in \Cref{fig:counterexample2} shows that such an approach does not work for \BEVS.

However, a quadratic kernel for \BEVS{} (or \BEOVS) can be achieved relatively easily.

\begin{theorem}
    \label{thm:kernel}
    \BEVS{} and \BEOVS{} each admit a kernel with at most~$6k(k+1)$ vertices and at most~$6k$
    critical independent sets, which can be computed in linear time.
\end{theorem}

\begin{proof}
    We first compute all critical independent sets and the critical independent set quotient graph~$\mathcal{I}$ of~$G$.
    Next, if a connected component in~$\mathcal{I}$ is an isolated vertex or two vertices connected by an edge, then we remove all vertices in the corresponding critical independent sets.

    We formalize this as the following reduction rule:

    \begin{redrule}
        \label{red:1}
        Remove all vertices contained in critical independent sets corresponding to isolated nodes or isolated edges in~$\mathcal{I}$.
    \end{redrule}
    Note that the former corresponds to a set of isolated vertices in~$G$ and the latter corresponds to a connected component that is a complete bipartite graph.
    Since both are bicliques and not connected to the rest of the graph, we can safely remove them.

    Next, we reduce the size of each critical independent set.
    \begin{redrule}
      \label{red:2}
      If~$|I_i| > k+1$ for some critical independent set~$I_i$, then remove an arbitrary vertex in~$I_i$ from~$G$.
    \end{redrule}
    We next prove that the reduction rule is safe.
    To this end, let~$G'$ be the graph obtained after removing a vertex~$u$ as described in the reduction rule.
    First assume that~$(G,k)$ is a yes-instance (of either problem).
    Since removing a vertex can never increase the distance from being a bicluster graph, $(G',k)$ is clearly also a yes-instance.
    Now assume that~$(G',k')$ is a yes-instance.
    Combining \Cref{lem:cis} with the observation that~$|I_i|$ remains strictly greater than $k$ after the removal of~$u$, we can assume without loss of generality that there is an optimal solution~$\sigma$, which does not add or remove any edges incident to vertices in~$I_i$ and also does not split any vertex in~$I_i$.
    Hence, adding the vertex~$u$ back does not change the fact that~$\sigma$ still results in a bicluster graph.

    We next analyze the running time.
    Computing all critical independent sets and the critical independent set quotient graph~$\mathcal{I}$ of~$G$ takes linear time~\cite{HM91}.
    Finding isolated vertices and edges in~$\mathcal{I}$ also takes linear time.
    Iterating over all critical independent sets and removing vertices also takes linear time.
    Finally, checking whether~$\mathcal{I}$ contains more than~$6k$ nodes after performing the above reduction rules exhaustively and returning a trivial no-instance if that is the case takes constant time.
    Since each step takes linear time, the entire kernel can be computed in linear time.

    It remains to show that if a graph~$G$ does not allow for the application of either of the two above reduction rules and contains more than~$6k$ critical independent sets, then~$(G,k)$ is a no-instance of \BEVS{} and of \BEOVS.
We show that if a graph does not allow for an application of either Reduction rules \ref{red:1} or \ref{red:2} and contains more than~$6k$ critical independent set, then we are dealing with a no-instance.
    Note that since \Cref{red:2} does not apply, this also shows that the number of vertices is upper bounded by~$6k(k+1)$.
    Afterwards, we analyze the running time of computing the kernel.

    Assume that there is a solution sequence~$\sigma$ of length at most~$k$ resulting in a bicluster graph~${H = (V',E')}$.
    We assume without loss of generality that~$\sigma$ satisfies \Cref{lem:noconversion,lem:cis}.
    We will partition the vertices in~$V'$ into two sets.
    The set~$X$ contains all vertices that are \emph{touched} by $\sigma$, that is, vertices that are the result of a vertex-split operation or vertices that are incident to an edge that is added or removed by~$\sigma$.
    Note that~$|X| \leq 2k$ as each edge is incident to two vertices and each vertex split adds two new vertices to the graph.
    The set~$Y$ contains all other vertices in~$V'$.
    Since vertices in~$Y$ are not split, they are also part of the original graph.
    Consider a connected component (a biclique)~$H_j=(A_j,B_j)$ in~$H$.
    Since \Cref{red:1} does not apply, at least one vertex in~$H_j$ is contained in~$X$.
    Moreover, each side of~$H_j$ contains vertices in~$Y$ from at most one critical independent set as shown next.
    Assume towards a contradiction that~$A_j$ contains vertices~$a_1,a_2 \in Y$ belonging to different critical independent sets.
    Since~$a_1,a_2 \in Y$, no edges incident to these two vertices are added or removed during~$\sigma$.
    Hence,~$B_j$ consists of exactly one copy of vertices in~$N(a_1)$ and the same for~$a_2$.
    However, these two neighborhoods are different as~$a_1$ and~$a_2$ belong to different critical independent sets (and are both contained in~$V_1$).
    Thus, $A_j$ only contains vertices in~$Y$ from one critical independent set and the same is true for~$B_j$ by the same argument.
    Hence, the number of critical independent sets with vertices in~$Y$ is at most~$2|X| \leq 4k$.
    Combined with the fact that~$|X| \leq 2k$ and therefore at most~$2k$ critical independent sets (in the original graph) contain vertices with copies in~$X$, this shows that the number of critical independent sets in the graph before applying~$\sigma$ is at most~$6k$.
    This concludes the proof.
\end{proof}

We mention that the question whether \BEVS{} and/or \BEOVS{} admit kernels with~$O(k)$ vertices is an interesting open problem for the future.

\section{An FPT Algorithm}
\label{sec:algorithm}

\Cref{thm:kernel} implies that \BEVS{} is fixed-parameter tractable since we can compute the kernel in linear time and then solve the kernel using brute force.
For \BEOVS, Abu-Khzam et al.~\cite{abu-khzam2025biclusterediting} developed an algorithm running in time~$O(5^{3k}k^{2k} n^2)$.
Unfortunately, we show that the algorithm contains a flaw and show a counter example where the algorithm does not compute an optimal solution.
Afterwards, we will present an algorithm that solves \BEVS{} and \BEOVS{} in~$O(k^{11k} + n + m)$ time.
This algorithm is an adaptation of an algorithm with similar running time due to Abu-Khzam et al.~\cite{AABD+23}.

The counter example for the algorithm by Abu-Khzam et al.~\cite{abu-khzam2025biclusterediting} is given in Figure~\ref{fig:counterexample3}.
In short,
the algorithm iteratively finds an induced~${P_4 = (a,b,c,d)}$, where~${a,c \in V_1}$ and~$c$ is not marked.
Then, it branches into five different cases: deleting any of the three edges, adding the edge~$\{a,d\}$, or marking~$c$.
Marking corresponds to splitting the vertex~$c$ but the authors delay the decision on how to split~$c$ until all edge modifications are guessed and then present an algorithm that optimally splits all marked vertices.
The problem is that marking vertices and only searching for~$P_4$s in which~$c$ is unmarked can make it impossible for the algorithm to delete an edge that has to be removed in any optimal solution.
We develop a different algorithm with a slightly worse running time in the parameter but better dependency on the input size for both \BEVS{} and \BEOVS.

\begin{figure}[h]
    \centering
    \begin{tikzpicture}
        \node[circle,draw,label=left:$A_1$,label=$1$] at (-1,6) (a1) {};
        \node[circle,draw,label=left:$A_2$,label=$1$] at (-1,5) (a2) {};
        \node[circle,draw,label=left:$A_3$,label=$1$] at (-1,4) (a3) {};
        \node[circle,draw,label=left:$A_4$,label=$1$] at (-1,3) (a4) {};
        \node[circle,draw,label=left:$A_5$,label=$5$] at (-1,2) (a5) {};
        \node[circle,draw,label=left:$A_6$,label=$5$] at (-1,1) (a6) {};
        \node[circle,draw,label=right:$B_1$,label=$5$] at (1,6) (b1) {} edge(a1) edge(a2);
        \node[circle,draw,label=right:$B_2$,label=$5$] at (1,5) (b2) {} edge(a2) edge(a3);
        \node[circle,draw,label=right:$B_3$,label=$1$] at (1,4) (b3) {} edge(a3) edge(a4);
        \node[circle,draw,label=right:$B_4$,label=$5$] at (1,3) (b4) {} edge(a4);
        \node[circle,draw,label=right:$B_5$,label=$5$] at (1,2) (b5) {} edge(a4) edge(a5);
        \node[circle,draw,label=right:$B_6$,label=$5$] at (1,1) (b6) {} edge(a3) edge(a6);
    \end{tikzpicture}
    \caption{An example of a graph in which the algorithm by Abu-Khzam et al.~\cite{abu-khzam2025biclusterediting} can fail to find an optimal solution.
    Each node represents a critical independent set and the number above it shows how many vertices are contained in it.
    The nodes on the left side represent the vertices in~$V_1$ and the nodes on the right represent~$V_2$.
    The only optimal solution (with~$k=4$) splits the vertices in~$A_2$, $A_3$, and~$A_4$ once each and deletes the edge between the vertex in~$A_3$ and in~$B_3$.
    Consider the case where the first~$P_4$ found by the algorithm contains vertices from~$A_1,B_1,A_2,$ and~$B_2$.
    The only guess to consider is to mark the vertex in~$A_2$.
    The second~$P_4$ found contain vertices from~$A_5,B_5,A_4$, and~$B_4$ and the vertex in~$A_4$ is marked.
    The third~$P_4$ contains vertices from~$A_6,B_6,A_3$, and~$B_2$.
    Only marking the vertex in~$A_3$ has to be considered.
    However, now no more~$P_4$ with an unmarked vertex~$c$ can be found and hence the algorithm computes a solution in which only vertices are split. The optimal way to do this splits the vertex in~$A_2$ once and each vertex in~$A_3 \cup A_4$ twice.
    Thus, the cost is~$5 > 4$ showing that the algorithm by Abu-Khzam et al.~\cite{abu-khzam2025biclusterediting} is flawed.}
    \label{fig:counterexample3}
\end{figure}

\begin{theorem}
    \BEVS{} and \BEOVS{} are solvable in ${O(k^{11k} + n + m)}$ time.
\end{theorem}

\begin{proof}
    First, we compute the kernel~$G'$ from \Cref{thm:kernel}, all critical independent sets, and the critical independent set quotient graph~$\mathcal{I}$ of the kernel in linear time \cite{HM91}.
    Note that~$\mathcal{I}$ contains at most~$6k$~vertices.
    Let~$\sigma$ be a solution that satisfies \Cref{lem:noconversion,lem:cis}.
Let~$\mathcal{X} = \{H_1, H_2, \ldots, H_\ell\}$, where each~$H_j = (A_j, B_j)$ is a biclique in~$\ver{G'}{\sigma}$.
Note also that~$\mathcal{X}$ contains~$\ell \leq 2k$ bicliques as each operation can complete at most two bicliques of the solution (removing an edge between two bicliques or splitting a vertex contained in both bicliques) and \Cref{red:1} removed all isolated bicliques (bicliques with no further connections in the input graph).
Hence, if there are more than~$2k$ bicliques in the solution, then we cannot reach the solution with~$k$ operations.
    To streamline the following argumentation, we will cover the nodes in~$\mathcal{I}$ by bicliques~$H_1,H_2,\ldots,H_{\ell=2k}$ and assume that an optimal solution contains exactly~$2k$ bicliques by allowing some of the bicliques to be empty.
    Next, let~$a_1,a_2,\ldots,a_{\ell},b_1,b_2,\ldots,b_{\ell},c$ be~$2\ell+1$ colors.
    Next, we iterate over all possible colorings of the nodes in~$\mathcal{I}$ such that each critical independent set~$I_i \subseteq V_1$ gets a color in~$\{c,a_1,a_2,\ldots,a_\ell\}$ and each critical independent set~$I_{i'} \subseteq V_2$ gets a color in~$\{c,b_1,b_2,\ldots,b_\ell\}$ for \BEVS{} and a color in~$\{b_1,b_2,\ldots,b_\ell\}$ for \BEOVS.
    Note that there are at most~$6k$ critical independent sets in~$G$ and hence there are at most~${(\ell+1)^{6k} \in O((2k+1)^{6k})}$ such colorings.

    The idea behind the coloring is the following.
    We will try to find a solution satisfying \Cref{lem:noconversion,lem:cis}.
    Each color~$a_j$ corresponds to the ``left side''~$A_j$ of a biclique~$H_j = (A_j,B_j)$ in the solution graph.
    The color~$b_j$ corresponds to set~$B_j$.
    That is, we try to find a solution where all (vertices in critical independent sets corresponding to) nodes of the same color (except for color~$c$) belong to the side of the same biclique in the solution.
    The color~$c$ indicates that the node will belong to multiple bicliques in the solution, that is, all vertices in the respective critical independent set will be split.
    Since each such split operation reduces~$k$ by one, we can reject any coloring in which the number of vertices in critical independent sets corresponding to nodes with color~$c$ is more than~$k$.
    In particular, we can reject any coloring in which more than~$k$ nodes have color~$c$.

    Next, we guess two indices~$i \in [k], j \in [\ell]$ and assume that the~$i$\textsuperscript{th} node of color~$c$ belongs\footnote{We only consider~$B_j$ in the case of \BEVS{} and the choice depends on whether the vertices in the critical independent set are contained in~$V_1$ or~$V_2$.}
    to~$A_j$ or~$B_j$---or that all nodes of color~$c$ have been assigned to all bicliques they belong to.
    Note that each guess is over~$k\ell+1$ possibilities.
    Moreover, at most~$k+1$ guesses do not reduce~$k$ by at least one (the last guess and the first time each index~$i \in [k]$ is guessed)
    Hence, we can make at most~$2k+1$ guesses.
    Thus, the number of such guesses is at most
    \[(k\ell+1)^{2k+1} = (2k^2+1)^{2k+1} \in O((2k+1)^{4k+1}).\]

    It remains to compute the best solution corresponding to each possible set of guesses.
    To this end, we first iterate over each pair of vertices and add an edge between them if this edge does not already exist and we guess that there is a biclique~$H_j$ which contains a copy of each of the two vertices.
    Moreover, we remove an existing edge between them if we guessed that the two vertices do not appear in a common biclique.
    Finally, we perform all split operations.
    Therein, we iteratively split one vertex~$v$ into two vertices~$u_1$ and~$u_2$ where~$u_1$ will be the vertex in some set~$A_j$ or~$B_j$ and~$u_2$ might be split further in the future.
    The vertex~$u_1$ is adjacent to all vertices that are guessed to belong to~$B_j$ or~$A_j$.
    The vertex~$u_2$ is adjacent to all vertices that~$u$ was adjacent to, except for vertices that are adjacent to~$u_1$ and not guessed to also belong to some other biclique~$H_{j'}$ which (some copy of)~$u_2$ belongs to.

    Since our algorithm performs an exhaustive search of all possible solutions satisfying \Cref{lem:noconversion,lem:cis}, it will find a solution if one exists.
    It only remains to analyze the running time.
    We first compute the kernel in~$O(n+m)$ time.
    We then try~$O((2k+1)^{6k})$ possible colorings of~$\mathcal{I}$ and for each coloring~$O{((2k+1)^{4k+1})}$~guesses.
    Afterwards, we compute the solution in~$O(k^2)$~time as~$n \in O(k)$ by \Cref{thm:kernel}.
    Thus, the overall running time is in
    \[{\hspace*{2.1cm} O((2k+1)^{10k+1} \cdot k^2 + n+m) \subseteq O(k^{11k}+n+m)}. \hspace*{2.1cm} \]
\end{proof}

\section{Conclusion}
\label{sec:conclusion}
We showed that both \BEVS{} and \BEOVS{} admit polynomial kernels with~$O(k^2)$ vertices computable in linear time.
This resolves two open problems by Abu-Khzam et al.~\cite{abu-khzam2025biclusterediting}.
Moreover, we show that their FPT algorithm for \BEOVS{} contains a flaw and present a different algorithm solving both problem variants in~${O(k^{11k} + n + m)}$~time.

We highlight several directions for future work.
One direction is approximation: is there a factor-$c$ approximation that can be computed in polynomial time for any constant~$c$ for either problem?
Another direction is to study \BEVS{} on non-bipartite input graphs.
A key question is whether a variant of \Cref{lem:cis} still holds in that setting.
Next, we leave it open whether the polynomial kernels can be improved to~$O(k)$ vertices.
Finally, we ask whether an algorithm running in~$2^{O(k)}\poly(n)$~time exists for \BEVS{} or \BEOVS.

\end{document}